\documentclass[11pt]{article}

\usepackage[english]{babel}
\usepackage{graphicx} 
\usepackage{amssymb,amsmath,amsfonts} 
\usepackage{stmaryrd} 
\usepackage[latin1]{inputenc} 
\usepackage[small,bf,hang]{caption} 
\usepackage{epic} 

\renewcommand{\atop}[2]{\genfrac{}{}{0pt}{}{#1}{#2}}
\newcommand{\Fser}[5]{\hspace{0cm} _{#1}F_{#2} \left( \left. \atop{#3}{#4} \right| #5 \right)}
\def\inprod#1{\left\langle #1 \right\rangle}
\def\abs#1{\left\lvert #1 \right\rvert}
\def\GZ#1{ \left| #1 \right> }

\newenvironment{proof}{\noindent \textbf{Proof:}}{$\hfill \Box$\\}

\newtheorem{theorem}{Theorem}
\newtheorem{proposition}[theorem]{Proposition}

\graphicspath{{Figuren/}}

\begin{document}

\begin{center}
  {\Large \bf Wigner quantization of some one-dimensional Hamiltonians}\\[5mm]
  {\bf G.~Regniers \footnote{E-mail: Gilles.Regniers@UGent.be}, }
  {\bf and J.\ Van der Jeugt} \footnote{E-mail: Joris.VanderJeugt@UGent.be} \\[1mm]
  Department of Applied Mathematics and Computer Science, Ghent University, \\
  Krijgslaan 281-S9, B-9000 Gent, Belgium.
\end{center}

\begin{abstract}
  \noindent Recently, several papers have been dedicated to the Wigner quantization of different Hamiltonians. In these examples, many interesting mathematical and physical properties have been shown. Among those we have the ubiquitous relation with Lie superalgebras and their representations. In this paper, we study two one-dimensional Hamiltonians for which the Wigner quantization is related with the orthosymplectic Lie superalgebra $\mathfrak{osp}(1|2)$. One of them, the Hamiltonian $\hat{H} = \hat{x} \hat{p}$, is popular due to its connection with the Riemann zeros, discovered by Berry and Keating on the one hand and Connes on the other. The Hamiltonian of the free particle, $\hat{H}_f = \hat{p}^2/2$, is the second Hamiltonian we will examine. Wigner quantization introduces an extra representation parameter for both of these Hamiltonians. Canonical quantization is recovered by restricting to a specific representation of the Lie superalgebra $\mathfrak{osp}(1|2)$.
\end{abstract}

\setcounter{equation}{0}

\section{Introduction}
Wigner quantization was initiated in 1950 by Eugene Wigner~\cite{Wigner} and gained in interest by physicists and mathematicians throughout the years. The principle of Wigner quantization is to dissociate oneself from the standard canonical commutation relations. In other words, for one-dimensional systems the canonical commutation relation 
\begin{equation} \label{CCR}
  [\hat{x}, \hat{p}] = i \hbar
\end{equation}
no longer holds. Instead, the presupposed equivalence of Hamilton's equations and the Heisenberg equations will relate the position operator $\hat{x}$ and the momentum operator $\hat{p}$ in a more general way. The result is a set of so-called compatibility conditions. From the validity of \eqref{CCR} one can deduce the compatibility conditions. In this sense, Wigner quantization is a more general approach than the widely known canonical quantization, where the canonical commutation relation \eqref{CCR} is assumed to be valid.

The compatibility conditions generally lead to triple relations containing commutators and anti-commutators. Only when the theory of Lie superalgebras was developed, it became clear that they are very useful in finding solutions for these triple relations. Therefore, the notion of Wigner quantum systems was only introduced much later than Wigner's original paper, by Palev~\cite{Palev-79, Palev-82, Palev-86}.

Wigner quantization has been studied extensively before~\cite{Manko-97, Horzela-99, Horzela-2000, Kapuscik-2000}. The $n$-dimensional harmonic oscillator has proved to be fertile ground for investigation in Wigner quantization as well~\cite{Palev-94, Palev-97, King-2003a, King-2003b, Palev-2006}. As mentioned, the compatibility conditions can usually be rewritten as defining relations of a Lie superalgebra. For example, the systems discussed in~\cite{LSVdJ-06, LSVdJ-08-1, RVdJ-09-2} have solutions related to $\mathfrak{gl}(1|n)$ and $\mathfrak{osp}(1|2n)$. In the present text, Wigner quantization of $\hat{H} = \hat{x} \hat{p}$ and $\hat{H}_f = \hat{p}^2/2$ will lead to the Lie superalgebra $\mathfrak{osp}(1|2)$. If we want to determine the spectrum of these Hamiltonians, we must study the representations of $\mathfrak{osp}(1|2)$. All irreducible $\ast$-representations of $\mathfrak{osp}(1|2)$ are known and have been classified in~\cite{RVdJ-10-Varna}. For the $\ast$-structure relevant in this paper, there appears to be only one such class of representations, the positive discrete series representations. They are characterized by a positive parameter $a$. The specific value $a=1/2$ will receive a lot of attention, as it represents the canonical picture. In this way, we will be able to verify if our generalized results are consistent with the known canonical case. 

The paper is roughly divided in two big parts. In the first part we discuss the Wigner quantization of a Hermitian operator corresponding to $\hat{x} \hat{p}$, namely $\hat{H}_b = (\hat{x}\hat{p}+\hat{p}\hat{x})/2$, and in the second part we handle the Hamiltonian of the free particle. At first sight it seems odd to deal with the less known Hamiltonian first. However, the calculations are more accessible for $\hat{x} \hat{p}$ than for the free particle. Therefore, we treat the former, less difficult case in detail and we will be more succinct throughout the second part of the paper.

In Section \ref{sec-Wigner-Hxp} we show how Wigner quantization of $\hat{x} \hat{p}$ works and how it can be connected to the Lie superalgebra $\mathfrak{osp}(1|2)$. In order to know how $\hat{H}_b$, $\hat{x}$ and $\hat{p}$ act as operators, we give a classification of all irreducible $\ast$-representations of $\mathfrak{osp}(1|2)$ in Section \ref{sec_representations}. At this point, we dispose of the actions of all relevant operators in any representation of $\mathfrak{osp}(1|2)$, so we could already compute the action of $\hat{H}_b$, $\hat{x}$ and $\hat{p}$ in order to find their spectrum. However, as it turns out these actions are related to certain orthogonal polynomials. Therefore, we first present some results on Meixner-Pollaczek polynomials, Laguerre polynomials and generalized Hermite polynomials in Section \ref{sec-polynomials}. After this, we are able to determine the spectra of $\hat{H}_b$, $\hat{x}$ and $\hat{p}$ in a standard way, shown in Section \ref{sec-spectrum}. We then compute the wave functions of the system in Section \ref{sec-generalized-wave-function}. Since Wigner quantization can be seen as a more general approach than canonical quantization, we can speak of generalized wave functions. Finally, we go back to the comfort zone of canonical quantization by taking $a=1/2$. Our results prove to be compatible with what is known from the canonical setting. 

The section concerning the free particle is roughly built up in the same way. However, the computations there are less detailed and the focus lies on the results.

\section{The Berry-Keating-Connes Hamiltonian $\hat{H} = \hat{x} \hat{p}$}
The recent popularity of the Hamiltonian $\hat{H} = \hat{x} \hat{p}$ must be attributed to its possible connection with the Riemann hypothesis~\cite{Sierra-2008, Rosu-2003, Aneva-1999}. This conjecture states that the non-trivial zeros of the Riemann zeta function can all be written as $1/2 + i t_n$, where the $t_n$ are real numbers. The first speculations about a potential relation between the Hamiltonian $\hat{H}$ and the Riemann zeros have been made in 1999, by Berry and Keating on the one hand~\cite{Berry-1999a, Berry-1999b} and Connes on the other hand~\cite{Connes-1999}. However, the idea of linking a certain Hamiltonian to the Riemann hypothesis is much older. \\
The origin of this suggestion lies almost a century behind us, when P\'olya proposed that the imaginary parts of the non-trivial Riemann zeros could correspond to the (real) eigenvalues of some self-adjoint operator. This statement is known as the Hilbert-P\'olya conjecture, although Hilbert's contribution to this is unclear. For a long time the Hilbert-P\'olya conjecture was regarded as a bold speculation, but it gained in credibility due to papers by Selberg~\cite{Selberg-1956} and Montgomery~\cite{Montgomery-1973}. \\
The historical commotion around the Hamiltonian $\hat{H}_b = \hat{x} \hat{p}$ inspired us to perform the Wigner quantization of this one-dimensional system. As we shall see, a lot of interesting results emerge.

\subsection{Wigner quantization and $\mathfrak{osp}(1|2)$ solutions} \label{sec-Wigner-Hxp}
In this section, we will perform the Wigner quantization of the Hamiltonian $\hat{H} = \hat{x} \hat{p}$. The easiest Hermitian operator corresponding to $\hat{x} \hat{p}$ is
\begin{equation}
  \hat{H}_b = \frac{1}{2} (\hat{x}\hat{p}+\hat{p}\hat{x}).
\end{equation}
In the procedure of Wigner quantization, one starts from the operator form of Hamilton's equations and the equations of Heisenberg. The former take the explicit form
\[
  \dot{\hat{x}} = \frac{\partial \hat{H}_b}{\partial \hat{p}} = \hat{x},  \qquad
  \dot{\hat{p}} = -\frac{\partial \hat{H}_b}{\partial \hat{x}} = -\hat{p}
\]
and the equations of Heisenberg can be written as
\[
  \dot{\hat{x}} = \frac{i}{\hbar} [\hat{H}_b,\hat{x}], \qquad
  \dot{\hat{p}} = \frac{i}{\hbar} [\hat{H}_b,\hat{p}].
\]
One then expresses the compatibility between these operator equations, thus creating a pair of compatibility conditions. For simplicity of notation, we set $\hbar=1$ and we find
\[
  [\hat{H}_b,\hat{x}] = -i\hat{x}, \qquad [\hat{H}_b,\hat{p}] = i\hat{p},
\]
or equivalently
\begin{equation} \label{CCs}
  [\{ \hat{x}, \hat{p} \},\hat{x}] = -2i\hat{x}, \qquad [\{ \hat{x}, \hat{p} \},\hat{p}] = 2i\hat{p},
\end{equation}
where $\{ \hat{x}, \hat{p} \}$ denotes the anticommutator between $\hat{x}$ and $\hat{p}$: $\{ \hat{x}, \hat{p} \} = \hat{x}\hat{p} + \hat{p}\hat{x}$.

The goal is now to find Hilbert spaces in which the operators $\hat{x}$ and $\hat{p}$ act as self-adjoint operators in such a way that they satisfy the compatibility conditions \eqref{CCs}. In our future calculations, we are not allowed to make any assumptions on the commutation relations between $\hat{x}$ and $\hat{p}$. The reader is urged to verify that all these computations take this restriction into account. The strategy will be to identify the algebra generated by $\hat{x}$ and $\hat{p}$ subject to \eqref{CCs}. Since the relations \eqref{CCs} contain anticommutators, Lie superalgebras come into the picture. In particular our problem will be connected to $\mathfrak{osp}(1|2)$ and representations of this Lie superalgebra will give us the demanded Hilbert spaces. 

\begin{proposition}
  The operators $\hat{x}$ and $\hat{p}$, subject to the relations \eqref{CCs}, generate the Lie superalgebra $\mathfrak{osp}(1|2)$.
\end{proposition}

\begin{proof}
Let us define new operators $b^+$ and $b^-$:
\begin{equation} \label{b+b-}
  b^\pm = \frac{\hat{x} \mp i\hat{p}}{\sqrt{2}}.
\end{equation}
These operators should satisfy $(b^\pm)^\dagger = b^\mp$, where the dagger operation stands for the ordinary Hermitian conjugate. In terms of the $b^\pm$, the operators $\hat{H}_b$, $\hat{x}$ and $\hat{p}$ take the form
\[
  \hat{H}_b = \frac{i}{2} \left( (b^+)^2 - (b^-)^2 \right), \qquad
  \hat{x}   = \frac{b^+ + b^-}{\sqrt{2}},                   \qquad
  \hat{p}   = \frac{i(b^+ - b^-)}{\sqrt{2}}.
\]
The compatibility conditions \eqref{CCs} are equivalent to the equations $[\hat{H}_b, b^\pm] = -i b^\mp$, which in turn can be written as
\[
  [(b^-)^2, b^+] = 2b^- \qquad \mbox{ and } \qquad [(b^+)^2, b^-] = -2b^+,
\]        
using the previous expression of $\hat{H}_b$ in terms of $b^+$ and $b^-$. By writing down the commutators and anticommutators explicitly, one sees that the latter two relations are equivalent to
\begin{equation} \label{osp_def_rel}
  \left[ \{b^-, b^+\}, b^\pm \right] = \pm 2b^\pm.
\end{equation}
These equations are known; they are the defining relations of the Lie superalgebra $\mathfrak{osp}(1|2)$~\cite{Ganchev}. \\
\end{proof}

For now, we only need to know that $\mathfrak{osp}(1|2)$ is generated by the odd elements $b^+$ and $b^-$, subject to the relations \eqref{osp_def_rel}. The even elements of $\mathfrak{osp}(1|2)$ are $h$, $e$ and $f$, defined by
\begin{equation} \label{def_hef}
  h = \frac{1}{2} \, \{ b^+, b^- \}, \qquad
  e = \frac{1}{2} \, (b^+)^2, \qquad
  f = - \frac{1}{2} \, (b^-)^2,
\end{equation}
satisfying the commutation relations
\[
  [h,e]=2e, \qquad [h,f]=-2f, \qquad [e,f]=h.
\]
If we regard the operators $\hat{H}_b$, $\hat{x}$ and $\hat{p}$ as operators on a certain representation space of $\mathfrak{osp}(1|2)$, we see that the Hamiltonian $\hat{H}_b$ is an even operator that can be written as
\[
  \hat{H}_b = i(e+f).
\]
It is also clear that the position and momentum operators are odd operators on this representation space.

We have now restated the crucial operators in terms of elements of the Lie superalgebra $\mathfrak{osp}(1|2)$. One of our main objectives, however, is to determine the spectrum of the operators $\hat{H}_b$, $\hat{x}$ and $\hat{p}$. In order to achieve this goal, we need to have a Hilbert space in which we can determine how the operators act. Lie superalgebra representations provide us with such a useful framework. In particular, we will work with irreducible $\ast$-representations of $\mathfrak{osp}(1|2)$.

\subsection{Irreducible $\ast$-representations of $\mathfrak{osp}(1|2)$} \label{sec_representations}
The $\mathbb{Z}_2$-graded algebra $\mathfrak{osp}(1|2)$, shortly described at the end of section \ref{sec-Wigner-Hxp} can be equipped with a so-called $\ast$-structure. This means that there exists an anti-linear anti-multiplicative involution $X \mapsto X^\ast$. So for $X, Y \in \mathfrak{osp}(1|2)$ and $a, b \in \mathbb{C}$ we have that $(aX+bY)^\ast = \bar{a}X^\ast + \bar{b}Y^\ast$ and $(XY)^\ast = Y^\ast X^\ast$. In our case, this $\ast$-structure is provided by the action $X \mapsto X^\dagger$. Restricting the $\ast$-structure to the even subalgebra of $\mathfrak{osp}(1|2n)$ gives us the Lie algebra $\mathfrak{su}(1,1)$. So $\mathfrak{su}(1,1)$ is generated by $h$, $e$ and $f$, defined in \eqref{def_hef}, which are related by the $\ast$-operation as follows:
\[
  h^\dagger = h, \qquad e^\dagger = -f, \qquad f^\dagger = -e.
\]
We will be working with representations of $\ast$-algebras. If the inner product that is defined on the representation space satisfies
\[
  \inprod{\pi(X)v, w} = \inprod{v, \pi(X^\ast)w},
\]
these representations are called $\ast$-representations. We emphasize that we choose this inner product to be antilinear in the first component. By doing so, we follow the convention accepted among physicists. \\
All irreducible $\ast$-representations of $\mathfrak{osp}(1|2)$ have been classified before. In 1981 this was done by Hughes~\cite{Hughes-81} using the shift operator technique, and much later by Regniers and Van der Jeugt~\cite{RVdJ-10-Varna} who also showed that two distinct classes of representations in~\cite{Hughes-81} are in fact equivalent. We summarize the results of the latter reference in order to have a clear image of all irreducible $\ast$-representations of $\mathfrak{osp}(1|2)$.
It turns out that there is only one such class of representations compatible with the $\ast$-condition $\bigl( b^\pm \bigr)^\ast = b^\mp$, characterized by a positive parameter $a$. These representations $\rho_a$ are called the positive discrete series representations and map the elements of $\mathfrak{osp}(1|2)$ to operators in $\ell^2(\mathbb{Z}_+)$. The representation space $V = \ell^2(\mathbb{Z}_+)$ is the Hilbert space of all square summable complex sequences. It has standard basis vectors $e_k, (k=0,1,2, \ldots)$, which are orthonormal with respect to the inner product on $V$. The parameter $a$ characterizing the representation is defined by
\[
  \rho_a(h) e_0 = a e_0.
\]
So $e_0$ is assumed to be an eigenvector of $\rho_a(h)$ with eigenvalue $a$. It is possible to show that, under this assumption, the representation space $V$ can be written as
\[
  V  = V_0 \oplus V_1.
\]
Separately, $V_0$ and $V_1$ are both lowest weight representation spaces of $\mathfrak{su}(1,1) \subset \mathfrak{osp}(1|2)$. In the remaining of this text, $V_0$ and $V_1$ will be referred to as the even and odd subspace respectively, and their basis elements are $e_{2n}$ and $e_{2n+1}$, with $n \geq 0$. This information can be thrown into a picture as follows

\vspace{1.0cm}

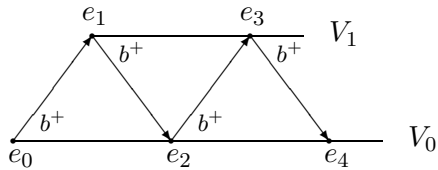
\begin{figure}[ht] 
\setlength{\unitlength}{0.7mm}
\begin{center}
  \begin{picture}(80,1)
    \drawline(15,0)(55,0)
    \drawline(0,-20)(70,-20)
    \put(60,-1){$V_1$}
    \put(75,-21){$V_0$}
    \multiput(0,-20)(30,0){3}{\circle*{1}}
    \put(-1,-24){$e_0$}
    \put(29,-24){$e_2$}
    \put(59,-24){$e_4$}
    \multiput(15,0)(30,0){2}{\circle*{1}}
    \put(13,3){$e_1$}
    \put(43,3){$e_3$}
    \multiput(15,0)(30,0){2}{\vector(3,-4){15}}
     \multiput(5,-18)(30,0){2}{\footnotesize{$b^+$}}
    \multiput(0,-20)(30,0){2}{\vector(3,4){15}}
     \multiput(20,-5)(30,0){2}{\footnotesize{$b^+$}}
  \end{picture}  
\end{center}
\vspace{1.5cm}
\caption{The representation space $V = V_0 \oplus V_1$}
\label{fig_rep_space}
\end{figure}
\noindent In the context of parity-specific terminology, we mention that $\rho_a(h)$, $\rho_a(e)$ and $\rho_a(f)$ are called even operators because their action is confined to one or the other subspace $V_0$ or $V_1$. Similarly, $\rho_a(b^+)$ and $\rho_a(b^-)$ are odd operators since they map from one subspace onto the other. This can be seen in Figure \ref{fig_rep_space}. Such an intuitive picture is helpful in order to understand one of the main parts of our paper: determining the spectrum of the operators $\hat{H}_b$, $\hat{x}$, $\hat{p}$ and later on $\hat{H}_f$. For mathematical arguments behind this classification, we refer to~\cite{RVdJ-10-Varna}. \\
Now the action of $b^+$ and $b^-$ must be determined for both representation spaces separately. We shall identify the representation space with an $\mathfrak{osp}(1|2)$ module and write $\rho_a(X) v = Xv$ from now on. Then we have~\cite{RVdJ-10-Varna}
\begin{equation} \label{action_b+b-}
  \begin{array}{rcl} 
    \rho_a(b^+) e_{2n}   & = & b^+ e_{2n} = \sqrt{2(n+a)} e_{2n+1} \\
    \rho_a(b^-) e_{2n}   & = & b^- e_{2n} = \sqrt{2n} \, e_{2n-1}  \\
    \rho_a(b^+) e_{2n+1} & = & b^+ e_{2n+1} = \sqrt{2(n+1)} e_{2n+2} \\
    \rho_a(b^-) e_{2n+1} & = & b^- e_{2n+1} = \sqrt{2(n+a)} e_{2n} 
  \end{array}
\end{equation}
We leave it to the reader to calculate the action of $h$, $e$ and $f$ on both subspaces with the help of their definitions \eqref{def_hef}. It is interesting to give the action of $[\hat{x}, \hat{p}]$:
\[
  [\hat{x}, \hat{p}] \, e_{2n} = 2ai \, e_{2n},        \qquad
  [\hat{x}, \hat{p}] \, e_{2n+1} = 2(1-a)i \, e_{2n+1}.
\]
We are left with the canonical commutation relation \eqref{CCR}, for $\hbar=1$, when $a$ takes the value $1/2$. This is one way to see that this represents the canonical case. \\
Now that we know the nature of the Hilbert spaces in which our operators act, we can try to find eigenvectors of $\hat{H}_b$, $\hat{x}$ and $\hat{p}$. In section \ref{sec-spectrum} we will see that these eigenvectors, and thus the spectrum of the respective operators, are related to orthogonal polynomials. Therefore, we need to introduce the relevant orthogonal polynomials and some of their most interesting properties.

\subsection{Some results on orthogonal polynomials} \label{sec-polynomials}
The most important formulas on all standard orthogonal polynomials are collected in the Askey-scheme~\cite{Koekoek, Koekoek2}. In this section, we will gather modifications and combinations of some formulas regarding the orthogonal polynomials that will arise in the following sections. The polynomials in question are Meixner-Pollaczek polynomials, Laguerre polynomials and generalized Hermite polynomials.

\subsubsection{Meixner-Pollaczek polynomials} 
The classical Meixner-Pollaczek polynomials~\cite{Koekoek, Koekoek2} are defined by
\[
  P_n^{(\lambda)}(x; \phi) = \frac{(2 \lambda)_n}{n!} \, e^{in \phi} \,
                             _2F_1 \left( \left. \atop{-n, \, \lambda + ix}{2 \lambda} \right|
                                                   1-e^{-2i \phi} \right),
\]
where $_2F_1$ denotes the usual hypergeometric series and $(a)_n = a(a-1) \ldots (a-n+1)$ is the Pochhammer symbol. As all orthogonal polynomials, the Meixner-Pollaczek polynomials satisfy a certain orthogonality relation and a three term recurrence relation. Both formulas are quite involved and can be found in~\cite{Koekoek}. Instead, we will use the notation $P_n(E)$ for $P_n^{(\frac{a}{2})}(E; \frac{\pi}{2})$, so we have
\[
  P_n(E) = i^n \frac{(a)_n}{n!} \, _2F_1 \left( \left. \atop{-n, \frac{a}{2} + iE}{a} \right| 2 \right).
\]
For these specific polynomials, the orthogonality relation takes the form
\begin{equation} \label{orthog_MP}
    \frac{1}{2 \pi} \int_{-\infty}^{+\infty} 
                    \lvert \, \Gamma \left( \frac{a}{2} + iE \right) \rvert^2  P_m(E) \, P_n(E) dE 
  = \frac{\Gamma(n+a)}{2^a \, n!} \, \delta_{mn},
\end{equation}
and the recurrence formula is
\[
  2E \, P_n(E) = (n+a-1) \, P_{n-1}(E) + (n+1) \, P_{n+1}(E).
\]
We need a normalized version of the Meixner-Pollaczek polynomials, which we will denote by $\tilde{P}_n(E)$:
\[
  \tilde{P}_n(E) = \frac{\lvert \Gamma(\frac{a}{2} + iE) \rvert}{2} \,
                   \sqrt{ \frac{2^a \, n!}{\pi \, \Gamma(n+a)} } P_n(E).
\]
Actually, we should speak about the pseudo-normalized Meixner-Pollaczek functions. For reasons that will become clear later, we like them to satisfy the following orthogonality relation:
\begin{equation} \label{orthog_rel_MP}
  \int_{-\infty}^{+\infty} \tilde{P}_m(E) \, \tilde{P}_n(E) dE = \frac{1}{2} \, \delta_{mn}.
\end{equation}
Of course, the recurrence relation is also different for the $\tilde{P}_n(E)$. It is given by
\begin{equation} \label{rec_normal_MP}
  2E \, \tilde{P}_n(E) = \sqrt{n(n+a-1)} \, \tilde{P}_{n-1}(E) + \sqrt{(n+1)(n+a)} \, \tilde{P}_{n+1}(E).
\end{equation}
This equation will come back in a different context in section \ref{subsec_spectrum_H}.

\subsubsection{Laguerre polynomials}
Laguerre polynomials will only be important for the free particle, but since they are also needed to define the generalized Hermite polynomials, we already summarize the main results for Laguerre polynomials at this moment. The Laguerre polynomials~\cite{Koekoek, Koekoek2} are defined by
\begin{equation} \label{laguerre-pol}
  L_n^{(\alpha)}(x) = \frac{(\alpha+1)_n}{n!} \,
                      _1F_1 \left( \left. \atop{-n}{\alpha + 1} \right| x \right).
\end{equation}
The orthogonality of the Laguerre polynomials is expressed by
\[
    \int_0^\infty e^{-x} x^\alpha L_m^{(\alpha)}(x) \, L_n^{(\alpha)}(x) dx 
  = \frac{\Gamma(n+\alpha+1)}{n!} \delta_{mn}, \qquad \alpha > -1.
\]
We define an altered version of these polynomials by
\[
  \tilde{L}_n^{(\alpha)}(x) = \sqrt{2} \, e^{-\frac{x}{2}} x^{\frac{\alpha}{2}}
                              \sqrt{\frac{ n!}{\Gamma(n+ \alpha +1)}} L_n^{(\alpha)}(x).
\]
Again, these polynomials are not really the normalized Laguerre polynomials. A factor of $\sqrt{2}$ is added in order to obtain
\[
  \int_0^\infty \tilde{L}_m^{(\alpha)}(x) \, \tilde{L}_n^{(\alpha)}(x) dx 
  = 2 \, \delta_{mn}.
\]
The three term recurrence relation of these functions can be written as
\begin{equation} \label{rec_Laguerre}
    (2n-x+\alpha+1) \, \tilde{L}_n^{(\alpha)}(x) 
  = \sqrt{n(n+\alpha)} \, \tilde{L}_{n-1}^{(\alpha)}(x) + \sqrt{(n+1)(n+\alpha+1)} \, \tilde{L}_{n+1}^{(\alpha)}(x)
\end{equation}
and will be useful in the context of the free particle.

\subsubsection{Generalized Hermite polynomials}
The generalized Hermite polynomials $Q_n^{(a)}(x)$ form the third family of orthogonal polynomials that is of importance to us. They are not listed in the Askey-scheme, but they are closely related to the Laguerre polynomials $L_n^{(\alpha)}(x)$, defined previously by equation \eqref{laguerre-pol}. For positive $n$, the definition of the generalized Hermite polynomials~\cite{Chihara-78} is given by
\[
  Q_{2n}^{(a)}(x)   = (-1)^n L_n^{(a-1)}(x^2), \qquad
  Q_{2n+1}^{(a)}(x) = (-1)^n x L_n^{(a)}(x^2).
\]
The classical Hermite polynomials are found by choosing $a=1/2$. As mentioned in the introduction, this is an important value for $a$. We will show how this indeed corresponds to the case of canonical quantization. At this point however, the link with any type of quantization is not obvious at all, so let us not jump too far ahead. 

Of course, there is also a set of orthogonality relations and a pair of recurrence relations for the generalized Hermite polynomials. Instead of giving them here, we immediately define the normalized version of these polynomials:
\begin{eqnarray*}
  \tilde{Q}_{2n}^{(a)}(x)   & = & |x|^{a-\frac{1}{2}} \, e^{-\frac{x^2}{2}}
                                  \sqrt{\frac{n!}{\Gamma(n+a)}} \, Q_{2n}^{(a)}(x), \\
  \tilde{Q}_{2n+1}^{(a)}(x) & = & |x|^{a-\frac{1}{2}} \, e^{-\frac{x^2}{2}}
                                  \sqrt{\frac{n!}{\Gamma(n+a+1)}} \, Q_{2n+1}^{(a)}(x). 
\end{eqnarray*}
The functions $\tilde{Q}_n^{(a)}(x)$ satisfy the following orthogonality relation:
\[
  \int_{-\infty}^{+\infty} \tilde{Q}_m^{(a)}(x) \tilde{Q}_n^{(a)}(x) dx = \delta_{mn}.
\]
The pair of recurrence relations corresponding to the normalized generalized Hermite polynomials $\tilde{Q}_n^{(a)}(x)$ is given by
\begin{equation} \label{rec_Hermite}
  \begin{aligned}[c]
    x \, \tilde{Q}_{2n}^{(a)}(x)   & =   \sqrt{n} \, \tilde{Q}_{2n-1}^{(a)}(x) 
                                       + \sqrt{n+a} \, \tilde{Q}_{2n+1}^{(a)}(x), \\
    x \, \tilde{Q}_{2n+1}^{(a)}(x) & =   \sqrt{n+a} \, \tilde{Q}_{2n}^{(a)}(x) 
                                       + \sqrt{n+1} \, \tilde{Q}_{2n+2}^{(a)}(x).
  \end{aligned}
\end{equation}
Again, this recurrence relation is the primary result of this section because the same formulas will characterize different objects in section \ref{subsec_spectrum_x}. We will then be able to identify these objects with the normalized generalized Hermite polynomials $\tilde{Q}_n^{(a)}(x)$.

\subsection{Spectrum of the operators $\hat{H}_b$, $\hat{x}$ and $\hat{p}$} \label{sec-spectrum}
The first main goal of this paper is to determine the spectrum of the essential operators in the $\mathfrak{osp}(1|2)$ representation space $V = V_0 \oplus V_1$. The operators $\hat{x}$ and $\hat{p}$ have comparable expressions in terms of $\mathfrak{osp}(1|2)$ generators and can thus be handled in a similar way. The spectrum of $\hat{H}_b$ on the other hand is an entirely different issue. Since $\hat{H}_b = i(e+f)$ is an even operator, using the terminology of section \ref{sec_representations}, the spectrum of $\hat{H}_b$ can be considered in both subspaces $V_0$ or $V_1$ separately. The operator $\hat{x}$ is clearly an odd operator, so in this case one cannot look at both subspaces individually. \\
Determining the spectrum of the operators $\hat{H}_b$ and $\hat{x}$ is done using the same method. One starts by defining a formal eigenvector of $\hat{H}_b$ and $\hat{x}$, with respective eigenvalues $E$ and $x$. These eigenvectors are unknown linear combinations of the basis vectors $e_n$ of $V$. A three term recurrence relation for the coefficients can then be calculated, and this will enable us to identify these coefficients with orthogonal polynomials. From the spectral theorem for unbounded self-adjoint operators~\cite{Dunford-63, Berezanskii-68} one can then derive that the spectrum of the operators is equal to the support of the weight function of the corresponding orthogonal polynomials. It is then easy to find the spectrum of $\hat{H}_b$ and $\hat{x}$. \\

\subsubsection{Spectrum of $\hat{H}_b = i(e+f)$} \label{subsec_spectrum_H}
For the spectrum of the Hamiltonian $\hat{H}_b$, we start by defining a formal eigenvector $u_0(E)$, for the eigenvalue $E$:
\begin{equation} \label{eigenvector_H}
  u_0(E) = \sum_{n=0}^\infty \alpha_{2n}(E) e_{2n},
\end{equation}
where $\alpha_{2n}(E)$ are some unknown functions to be determined. We have already argued that the spectrum of $\hat{H}_b$ can by considered in both subspaces $V_0$ and $V_1$ separately. The vector $u_0(E)$ is an element of the even subspace $V_0$ and the action of $\hat{H}_b$ on this vector results in another element of $V_0$, as one sees from
\[
  \hat{H}_b \, u_0(E) = \sum_{n=0}^\infty \alpha_{2n}(E) 
                        \left( i \sqrt{(n+1)(n+a)} \, e_{2n+2} - i \sqrt{n(n+a-1)} \, e_{2n-2} \right).
\]
This has been established by rewriting $\hat{H}_b=i(e+f)$ with the help of \eqref{def_hef}, which makes it possible to use equations \eqref{action_b+b-} in order to determine the action of $\hat{H}_b$ on $u_0(E)$. We can try to find the unknown coefficients $\alpha_{2n}(E)$ by collecting the coefficients of $e_{2n}$ in the expression $\hat{H}_b u_0(E) = E u_0(E)$. We end up with the following recurrence relation:
\[
  E \, \alpha_{2n} = i \sqrt{n(n+a-1)} \, \alpha_{2n-2} - i \sqrt{(n+1)(n+a)} \, \alpha_{2n+2}.
\]
We then define $\tilde{\alpha}_{2n}(E) = (-i)^n \alpha_{2n}(E)$, for which the recurrence relation reads
\[
  E \, \tilde{\alpha}_{2n} =   \sqrt{n(n+a-1)} \, \tilde{\alpha}_{2n-2} 
                             + \sqrt{(n+1)(n+a)} \, \tilde{\alpha}_{2n+2}.
\]
If we compare this with equation \eqref{rec_normal_MP}, we see that both formulas are almost the same. The unknown objects $\tilde{\alpha}_{2n}(E)$ can therefore be identified with the normalized Meixner-Pollaczek polynomials. We have
\[
  \tilde{\alpha}_{2n}(E) = \tilde{P}_n( \frac{E}{2} ).
\]
We have found the coefficients $\alpha_{2n}(E)$ which determine the formal eigenvector \eqref{eigenvector_H}. They can be written as 
\begin{equation} \label{alpha_2n}
  \alpha_{2n}(E) = (-1)^n \sqrt{ \frac{(a)_n}{n!} } \,
                   {\cal{A}}_0(E) \,
                   _2F_1 \left( \left. \atop{-n, \frac{a+iE}{2}}{a} \right| 2 \right),
\end{equation}
with
\[
  {\cal{A}}_0(E) = \frac{ \abs{ \Gamma \left( \frac{a+iE}{2} \right) } }
                { \sqrt{2^{2-a} \, \pi \, \Gamma(a)} } \, .
\]
So far, we have used the language of formal eigenvectors, but the result can also be formulated in a different mathematical way. Herein, $\hat{H}_b$ acts as a Jacobi operator on the basis $\{ e_{2n} | n \in \mathbb{Z}_+ \}$ of $V_0 \cong \ell^2(\mathbb{Z}_+)$. Next, one defines a map $\Lambda$ from $\ell^2(\mathbb{Z}_+)$ to square integrable functions $L^2(\mathbb{R}, w(E)dE)$, where $w(E) = |\Gamma(\frac{a}{2} + iE)/2|$ is the weight function corresponding to Meixner-Pollaczek polynomials, by
\[
  (\Lambda e_{2n}) (E) = (-i)^n \sqrt{\frac{2^a \, n!}{\pi \, \Gamma(n+a)}} \, P_n(\frac{E}{2}).
\]
Then $\Lambda \circ \hat{H}_b = M_E \circ \Lambda$, i.e. $\Lambda$ intertwines $\hat{H}_b$ acting in $\ell^2(\mathbb{Z}_+)$ with the multiplication operator $M_E$ on $L^2(\mathbb{R}, w(E)dE)$, see~\cite[Prop. 3.1]{KoelinkVdJ-98}. Since we are really dealing with aspects of Wigner quantization in this paper, we shall not overload it with stricter terminology and just use the for physicists more familiar language of formal eigenvectors, delta-functions, etc.

We can now rely on the spectral theorem to find the spectrum of $\hat{H}_b$ in $V_0$: it is equal to the support of the weight function of the Meixner-Pollaczek polynomials. This weight function, accompanying the polynomials under the integral in equation \eqref{orthog_MP}, has the real axis as its support. As a result, the spectrum of $\hat{H}_b$ in $V_0$ is $\mathbb{R}$. \\
The same technique will give us the spectrum of $\hat{H}_b$ in $V_1$. We start by defining a formal eigenvector $u_1(E)$, determined by the coefficients $\alpha_{2n+1}(E)$. These are then calculated in the same way as for the coefficients $\alpha_{2n}(E)$. The analysis and results can be copied exactly, but the parameter $a$ has to be changed into $a+1$. Thus we have
\begin{equation} \label{alpha_2n+1}
  \alpha_{2n+1}(E) = (-1)^n \sqrt{ \frac{(a+1)_n}{n!} } \,
                     {\cal{A}}_1(E) \,
                     _2F_1 \left( \left. \atop{-n, \frac{a+1+iE}{2}}{a+1} \right| 2 \right),
\end{equation}
with
\[
  {\cal{A}}_1(E) = \frac{ \abs{ \Gamma \left( \frac{a+1+iE}{2} \right) } }
                { \sqrt{2^{1-a} \, \pi \, \Gamma(a+1)} } \, .
\]
The reason for this is that in $V_0$ the lowest weight vector is $e_0$ and $h \, e_0 = a \, e_0$ characterizes this representation. In $V_1$ on the other hand, $e_1$ is the lowest weight vector, and the corresponding lowest weight is $a+1$. The conclusion in this case is similar: the spectrum of $\hat{H}_b$ in $V_1$ is $\mathbb{R}$. \\
Combining these results, we have:

%
\begin{theorem} \label{th_specH}
  In the $\mathfrak{osp}(1|2)$ representation space $V$, the Hamiltonian $\hat{H}_b$ has formal eigenvectors
\[
  u_0(E) = \sum_{n=0}^\infty \alpha_{2n}(E) e_{2n},
\]
with coefficients determined by equation \eqref{alpha_2n}, and 
\[
  u_1(E) = \sum_{n=0}^\infty \alpha_{2n+1}(E) e_{2n+1},
\]
with coefficients \eqref{alpha_2n+1}. The spectrum of $\hat{H}_b$ in $V$ is $\mathbb{R}$ with multiplicity 2.
\end{theorem}
It is important to note that the coefficients $\alpha_n(E)$ have been chosen in such a way that the vectors $u_0(E)$ and $u_1(E)$ are delta function normalized vectors. To support this statement, we first observe that
\begin{equation} \label{inprod_alphas_even}
    \inprod{ \alpha_{2m}(E), \alpha_{2n}(E) } 
  = \int_{-\infty}^{+\infty} \alpha_{2m}^*(E) \, \alpha_{2n}(E) \, dE 
  = \delta_{mn},
\end{equation}
which is a direct consequence of equation \eqref{orthog_rel_MP}. We then multiply both sides of equation \eqref{eigenvector_H} by $\alpha_{2m}^*(E)$ and integrate. With the help of equation \eqref{inprod_alphas_even} this yields
\[
    \int_{- \infty}^{+ \infty} u_0(E) \, \alpha_{2m}^*(E) \, dE
  = \sum_{n=0}^\infty e_{2n} \int_{- \infty}^{+ \infty} \alpha_{2n}(E) \, \alpha_{2m}^*(E) \, dE
  = e_{2m}.
\]
Now we have an integral expression for $e_{2m}$, which is used to obtain
\begin{align*}
  u_0(E') & = \sum_{m=0}^\infty a_{2m}(E') \int_{- \infty}^{+ \infty} u_0(E) \, \alpha_{2m}^*(E) \, dE \\
          & = \int_{- \infty}^{+ \infty} \left( \sum_{m=0}^\infty \alpha_{2m}(E') \, \alpha_{2m}^*(E) \right) u_0(E) \, dE.
\end{align*}
The definition of the Dirac delta function implies that
\[
  \sum_{m=0}^\infty \alpha_{2m}^*(E) \, \alpha_{2m}(E') = \delta(E-E').
\]
From this, we can draw the immediate conclusion that $u_0(E)$ is a delta function normalized vector. Indeed:
\[
  \inprod{u_0(E), u_0(E')} = \sum_{n=0}^\infty \alpha_{2n}^*(E) \, \alpha_{2n}(E') = \delta(E-E').
\]
In an analogous way one can prove the same property for the vector $u_1(E)$, since equation \eqref{inprod_alphas_even} is also valid for odd indices. More generally, when a vector is decomposed in a certain orthonormal basis (in this case the $\mathfrak{osp}(1|2)$ representation space basis $e_n$), and when the coordinates (the coefficients $\alpha_n(E)$) are also orthonormal, then the vector is normalized with respect to the delta function. We will encounter this scenario for the eigenvectors of $\hat{x}$, $\hat{p}$ and $\hat{H}_f$.

\subsubsection{Spectrum of $\hat{x} = \frac{(b^+ + b^-)}{\sqrt{2}}$} \label{subsec_spectrum_x}
The spectrum of the operator $\hat{x}$ will be related to generalized Hermite polynomials, and the method of reaching this relation is quite similar as before. The biggest difference is that we now have to take the entire representation space $V$ into account. Thus, the formal eigenvector $v(x)$ of $\hat{x}$ is defined by
\begin{equation} \label{eigenvector_x}
  v(x) = \sum_{n=0}^\infty \beta_n(x) e_n,
\end{equation}
satisfying the relation $ \hat{x} \, v(x) = x \, v(x)$. The left hand side of this eigenvalue equation can be written as
\[
          \sum_{n=0}^\infty \beta_{2n}   (\sqrt{n} \, e_{2n-1} + \sqrt{n+a} \, e_{2n+1})
  \, + \, \sum_{n=0}^\infty \beta_{2n+1} (\sqrt{n+a} \, e_{2n} + \sqrt{n+1} \, e_{2n+2}),
\]
using $\hat{x} = \frac{(b^+ + b^-)}{\sqrt{2}}$ together with equations \eqref{action_b+b-}. We can then compare the coefficients of $e_n$ on both sides of the equation, which will result in a pair of recurrence relations. We have
\begin{eqnarray*}
  x \, \beta_{2n}   & = &   \sqrt{n} \, \beta_{2n-1} 
                          + \sqrt{n+a} \, \beta_{2n+1}, \\
  x \, \beta_{2n+1} & = &   \sqrt{n+a} \, \beta_{2n} 
                          + \sqrt{n+1} \, \beta_{2n+2}.
\end{eqnarray*}
It should come as no surprise that we recognize the pair of recurrence relations \eqref{rec_Hermite} of the normalized generalized Hermite polynomials. Therefore, we can identify the unknown coefficients $\beta_n(x)$ with those polynomials:
\[
  \beta_n(x) = \tilde{Q}_n^{(a)}(x).
\]
The explicit expression for the $\beta_{2n}(x)$ is given by
\begin{equation} \label{beta_2n}
  \beta_{2n}(x) = (-1)^n \sqrt{ \frac{(a)_n}{n!} } \,
                   {\cal{B}}_0(x) \, \,
                   _1F_1 \left( \left. \atop{-n}{a} \right| x^2 \right),
\end{equation}
with
\[
  {\cal{B}}_0(x) = \frac{\abs{x}^{a-\frac{1}{2}} e^{-\frac{x^2}{2}}}{\sqrt{\Gamma(a)}}.
\]
A similar expression describes the $\beta$'s with an odd index:
\begin{equation} \label{beta_2n+1}
  \beta_{2n+1}(x) = (-1)^n \sqrt{ \frac{(a+1)_n}{n!} } \,
                     {\cal{B}}_1(x) \,
                     _1F_1 \left( \left. \atop{-n}{a+1} \right| x^2 \right),
\end{equation}
with
\[
  {\cal{B}}_1(x) = \frac{x \abs{x}^{a-\frac{1}{2}} e^{-\frac{x^2}{2}}}{\sqrt{\Gamma(a+1)}}.
\]
The same arguments as for the spectrum of $\hat{H}_b$ allow us to determine the spectrum of $\hat{x}$. We have

%
\begin{theorem}
  In the $\mathfrak{osp}(1|2)$ representation space $V$, the position operator $\hat{x}$ has formal eigenvectors
\[
  v(x) = \sum_{n=0}^\infty \beta_n(x) e_n,
\]
with coefficients determined by equations \eqref{beta_2n} and \eqref{beta_2n+1}. The spectrum of $\hat{x}$ in $V$ is equal to $\mathbb{R}$, which is the support of the weight function of the generalized Hermite polynomials.
\end{theorem}
Note that the major difference with Theorem \ref{th_specH} lies in the fact that the spectrum of the position operator does not have a double multiplicity. Each eigenvalue $x$ belongs to exactly one eigenvector $v(x)$.

\subsubsection{Spectrum of $\hat{p} = \frac{i(b^+ - b^-)}{\sqrt{2}}$} \label{subsec_spectrum_p}
Determining the spectrum of $\hat{p}$ is now just a formality because of the similar expressions for $\hat{x}$ and $\hat{p}$. In fact, if we denote the formal eigenvector of $\hat{p}$ for the eigenvector $p$ by
\[
  w(p) = \sum_{n=0}^\infty \gamma_n(p) e_n,
\]
then it is not hard to see that the coefficients $\gamma_n(p)$ can be written as
\[
  \gamma_n(p) = i^n \beta_n(p),
\]
where the $\beta_n(p)$ are given by equations \eqref{beta_2n} and \eqref{beta_2n+1}. Thus, these coefficients are again generalized Hermite polynomials and we have
%
\begin{theorem}
  In the $\mathfrak{osp}(1|2)$ representation space $V$, the momentum operator $\hat{p}$ has formal eigenvectors
\[
  w(p) = \sum_{n=0}^\infty i^n \beta_n(p) e_n,
\]
with coefficients determined by equations \eqref{beta_2n} and \eqref{beta_2n+1}. Just like $\hat{x}$, $\hat{p}$ has a spectrum in $V$ that is equal to $\mathbb{R}$.
\end{theorem}

As pointed out in section \ref{subsec_spectrum_H}, the eigenvectors of $\hat{x}$ and $\hat{p}$ are normalized with respect to the delta function. Thus we have
\[
  \inprod{v(x), v(x')} = \delta(x-x')
\]
and
\[
  \inprod{w(p), w(p')} = \delta(p-p'),
\]
ensuing from the fact that $\beta_n(x)$ and $\gamma_n(p)$ are orthonormal functions.

\subsection{Generalized wave functions} \label{sec-generalized-wave-function}
In the previous section we have constructed formal eigenvectors for all relevant operators $\hat{H}_b$, $\hat{x}$ and $\hat{p}$. Finding the wave functions corresponding to the physical states $u_0(E)$, $u_1(E)$, $v(x)$ and $w(p)$ requires the mutual inner products between these vectors. Therefore, we will compute these inner products before analyzing the wave functions explicitly.

\subsubsection{Mutual inner products}
This section is dedicated to the purely mathematical calculation of all inner products between the vectors $u_0(E)$, $u_1(E)$, $v(x)$ and $w(p)$. Two existing formulas will be crucial in these calculations. One of them is found in~\cite[Proposition 2]{VdJ-Jagannathan-98} and states that
\begin{equation} \label{eq-Jagannathan}
\begin{array}{rcl}
  &   & \hspace{-2.0cm}
        \displaystyle{
        \sum_{n=0}^\infty \frac{(a)_n}{n!} \,
                          \Fser{2}{1}{-n, b}{a}{y} \Fser{1}{1}{-n}{a}{x} t^n
        }                                                                    \\
        \hspace{1.0cm}
  & = & \displaystyle{
        (1-t)^{b-a} (1-t+yt)^{-b} \, e^{\frac{xt}{t-1}} \, 
        \Fser{1}{1}{b}{a}{\frac{xyt}{(1-t)(1-t+yt)}},
        }
\end{array}
\end{equation}
where $|t|<1$. The other formula is the following approximation~\cite[equation (7.9)]{Temme-96}:
\begin{equation} \label{eq-Temme}
    \lim_{z \to \infty} \Fser{1}{1}{c}{a}{z} 
  = \lim_{z \to \infty} \frac{\Gamma(a)}{\Gamma(c)} \, e^z \, z^{c-a} 
    \sum_{n=0}^\infty \frac{(a-c)_n \, (1-c)_n}{n!} \, z^{-n}.
\end{equation}
These two equations will appear useful in the following.

Let us first examine the inner products between the eigenvectors of $\hat{x}$ and $\hat{H}_b$. We will only calculate
\[
  \inprod{ v(x), u_0(E) } = \sum_{n=0}^\infty \alpha_{2n}(E) \beta_{2n}^*(x)
\]
explicitly, since $\inprod{ v(x), u_1(E) }$ is found in a highly similar way. The functions $\beta_n(x)$ are real, so the complex conjugation can be dropped. Using equations \eqref{alpha_2n} and \eqref{beta_2n}, we can write $\inprod{ v(x), u_0(E) }$ as
\[
  {\cal{A}}_0(E) \, {\cal{B}}_0(x) 
  \sum_{n=0}^\infty \frac{(a)_n}{n!} \, 
                    \Fser{2}{1}{-n, \frac{a+iE}{2}}{a}{2}
                    \Fser{1}{1}{-n}{a}{x^2}.
\]
In order to determine this summation, we need to add a factor $t^n$ to the $n$th summand and then take the limit for $t \to 1$, with $t<1$. Using equation \eqref{eq-Jagannathan}, we rewrite
$\sum_{n=0}^\infty \alpha_{2n}(E) \beta_{2n}(x) t^n$ as
\begin{equation} \label{partial_result_inprod}
  {\cal{A}}_0(E) \, {\cal{B}}_0(x) 
  (1-t)^{\frac{-a+iE}{2}} \, (1+t)^{\frac{-a-iE}{2}} \, e^{\frac{x^2 t}{t-1}} \, 
  \Fser{1}{1}{\frac{a+iE}{2}}{a}{\frac{2x^2 t}{(1-t)(1+t)}}.
\end{equation}
The limit for $t \to 1$ is found with the help of equation \eqref{eq-Temme}. After simplification, we have
\begin{align}
  \inprod{ v(x), u_0(E) } & = {\cal{A}}_0(E) \, {\cal{B}}_0(x) \, \frac{\Gamma(a)}{\Gamma(\frac{a+iE}{2})} 
                                                     \, \lim_{t \to 1} \, e^\frac{x^2 t}{1+t} 
                                                     \, (\abs{x} \sqrt{2t})^{-a+iE} (1+t)^{-iE}      \nonumber            \\
                                & = \frac{ \abs{\Gamma(\frac{a+iE}{2})} }
                                         { \Gamma(\frac{a+iE}{2}) } \, 
                                    \frac{ \abs{x}^{iE-\frac{1}{2}} }
                                         { \sqrt{\pi \, 2^{iE+2}} }                                  \label{vx_v0E} \\
                                & = {\cal{C}}_0(E) \, \frac{|x|^{iE-\frac{1}{2}}}{2 \sqrt{\pi}},             \nonumber
\end{align}
where $|{\cal{C}}_0(E)|^2 = 1$.
%
%
A similar derivation provides us with the inner product $\inprod{v(x), u_1(E)}$. It is
\begin{align}
  \inprod{ v(x), u_1(E) } & = \frac{ \abs{\Gamma(\frac{a+1+iE}{2})} }
                                         { \Gamma(\frac{a+1+iE}{2}) } \, 
                                    \frac{ x \abs{x}^{iE-\frac{3}{2}} }
                                         { \sqrt{\pi \, 2^{iE+2}} }         \label{vx_v1E}           \\
                                & = {\cal{C}}_1(E) \, \frac{x|x|^{iE-\frac{3}{2}}}{2 \sqrt{\pi}},            \nonumber
\end{align}
where $|{\cal{C}}_1(E)|^2 = 1$.
%
%

We can deduce the inner products between the formal eigenvectors of $\hat{p}$ and $\hat{H}_b$ using \eqref{vx_v0E} and \eqref{vx_v1E}. First observe that
\[
  \Fser{2}{1}{a, b}{c}{z} = (1-z)^{-a} \Fser{2}{1}{a, c-b}{c}{\frac{z}{z-1}},
\]
together with ${\cal{A}}_0(E) = {\cal{A}}_0(-E)$, implies that
\[
  (-1)^n \alpha_{2n}(E) = \alpha_{2n}(-E).
\]
A similar identity holds for odd indices of the coefficients $\alpha_n(E)$: $(-1)^n \alpha_{2n+1}(E) = \alpha_{2n+1}(-E)$. We have used a standard hypergeometric identity that can be found in every book with an introduction to hypergeometric functions. We mention~\cite{Koekoek2} as an example. It is then possible to see that
\begin{equation} \label{wp-u0E}
  \inprod{w(p), u_0(E)} = \inprod{v(p), u_0(-E)}
\end{equation}
and
\begin{equation} \label{wp-u1E}
  \inprod{w(p), u_1(E)} = -i \inprod{v(p), u_1(-E)},
\end{equation}
where the inner products with a vector $v$ can be found in equations \eqref{vx_v0E} and \eqref{vx_v1E}.

What remains to be calculated is the inner product between the eigenvectors of the position and momentum operators $\hat{x}$ and $\hat{p}$. This inner product can be written as:
\[
  \inprod{v(x), w(p)} = \sum_{n=0}^\infty i^n \beta_n(x) \beta_n(p),
\]
which splits up as a sum over even indices and a sum over odd indices. Both series can be computed in an analogous fashion. For the even part we have
\begin{equation} \label{even-part-vx-vp}
    \sum_{n=0}^\infty (-1)^n \beta_{2n}(x) \beta_{2n}(p)
  = {\cal{B}}_0(x) {\cal{B}}_0(p) \sum_{n=0}^\infty (-1)^n \frac{(a)_n}{n!} \, \, \Fser{1}{1}{-n}{a}{x^2} \Fser{1}{1}{-n}{a}{p^2}.
\end{equation}
If we manage to write one of the $_1 F_1$-series as a $_2 F_1$ Gauss hypergeometric function, we can rely on equation \eqref{eq-Jagannathan} to simplify this equation. We shall make use of the following identity (see e.g.~\cite{Koekoek}):
\begin{equation} \label{lim-2F1}
  \lim_{b \to \infty} \, \Fser{2}{1}{a, b}{c}{\frac{z}{b}} = \, \Fser{1}{1}{a}{c}{z}.
\end{equation}
Applying this to $\, \Fser{1}{1}{-n}{a}{x^2}$, we see that the summation in the right-hand side of equation \eqref{even-part-vx-vp} can be calculated with the help of equation \eqref{eq-Jagannathan}. After some simplifications, this summation translates to
\[
  2^{-a} e^\frac{p^2}{2} \lim_{b \to \infty} \left( \frac{b}{b-\frac{x^2}{2}} \right)^{b} 
                                             \, \Fser{1}{1}{b}{a}{\frac{x^2 p^2}{2 (x^2-2b)}}.
\]
Now the limit for $b \to \infty$ can be taken for both factors. The first becomes $e^{x^2/2}$, while the second is $\Fser{0}{1}{-}{a}{-x^2 p^2 /4}$. Putting all this together, we find that the even part of the inner product $\inprod{v(x), w(p)}$ is
\[
    \sum_{n=0}^\infty (-1)^n \beta_{2n}(x) \beta_{2n}(p)
  = \frac{|xp|^{a-\frac{1}{2}}}{2^a \, \Gamma(a)} \, \, \Fser{0}{1}{-}{a}{-\frac{x^2 p^2}{4}}.
\]
The odd part of the inner product is found in the same way, so we have
\begin{equation} \label{vx_vp}
    \inprod{v(x), w(p)}
  = \frac{|xp|^{a-\frac{1}{2}}}{2^a \, \Gamma(a)} 
    \left( \Fser{0}{1}{-}{a}{-\frac{x^2 p^2}{4}} 
           + \frac{ixp}{2a} \, \, \Fser{0}{1}{-}{a+1}{-\frac{x^2 p^2}{4}} \right).
\end{equation}
In the canonical case, many of these expressions simplify significantly. We discuss these simplifications for $a=1/2$ in section \ref{subsec-canonical}.

\subsubsection{Generalized wave functions and the canonical case} \label{subsec-canonical}
Consider an arbitrary state of the system $\GZ{\psi}$, written in Dirac's bra-ket notation. Assume that the eigenstates of the position operator are denoted by $\GZ{x}$. Then the spatial wave function of the system is found by
\[
  \psi(x) = \inprod{x | \psi}.
\]
Similarly, an inner product describes the wave function in the momentum space:
\[
  \psi(p) = \inprod{p | \psi},
\]
where $\GZ{p}$ represents the momentum eigenstates.

We have written the position and momentum eigenvectors as $v(x)$ and $w(p)$ respectively. Their inner product, given by equation \eqref{vx_vp}, represents the wave function of the particle being located at position $x$, when the system is in the momentum eigenstate $p$. This result is compatible with the canonical case where $a=1/2$. In canonical quantization, $\hat{x}$ and $\hat{p}$ are known. The operator $\hat{x}$ is simply multiplication with $x$ and $\hat{p} = -i \partial_x$. For $a=1/2$ equation \eqref{vx_vp} reduces to
\[
  \inprod{v(x), w(p)} = \frac{1}{\sqrt{2} \, \Gamma(\frac{1}{2})} \bigl( \cos(xp) + i \sin(xp) \bigr)
                      = \frac{1}{\sqrt{2 \pi}} \, e^{ixp},
\]
which is an eigenfunction of the canonical interpretation of the operator $\hat{p}$ with eigenvalue $p$.

The case where the eigenstate $\GZ{\psi}$ corresponds to the energy $E$ needs to be handled with a little more care, for there are two energy eigenstates corresponding to $E$. Both $u_0(E)$ and $u_1(E)$ belong to the same energy eigenvalue, thus inducing two independent wave functions $\psi_E^0(x)$ and $\psi_E^1(x)$. The previous results \eqref{vx_v0E} and \eqref{vx_v1E} allow us to write
\begin{equation} \label{psi-E-0}
  \psi_E^0(x) = {\cal{C}}_0(E) \, \frac{|x|^{iE-\frac{1}{2}}}{2 \sqrt{\pi}}
\end{equation}
and
\begin{equation} \label{psi-E-1}
  \psi_E^1(x) = {\cal{C}}_1(E) \, \frac{x|x|^{iE-\frac{3}{2}}}{2 \sqrt{\pi}},
\end{equation}
with $|{\cal{C}}_0|^2 = |{\cal{C}}_1|^2 = 1$. Therefore the general wave function of the particle when the system's energy equals $E$ must be of the form
\begin{equation} \label{wave-function-Hb}
  \psi_E^{(a)}(x) = A \psi_E^0(x) + B \psi_E^1(x),
\end{equation}
with $A$ and $B$ complex coefficients satisfying $|A|^2 + |B|^2 = 1$. This result is compatible with the canonical case as well, which is no surprise since $\psi_E^{(a)}(x)$ is practically independent of $a$. In fact, in equations \eqref{psi-E-0} and \eqref{psi-E-1} $a$ only appears explicitly in the phase factors ${\cal{C}}_0(E)$ and ${\cal{C}}_1(E)$. For $\hat{x} = x$ and $\hat{p} = -i \partial_x$, the Hamiltonian $\hat{H}_b$ converts into
\[
  \hat{H}_b = -i(x \partial_x + \frac{1}{2}),
\]
which, for $a=1/2$, indeed has $\psi_E^0(x)$ and $\psi_E^1(x)$ as eigenfunctions with eigenvalue $E$. Moreover, the general wave function \eqref{wave-function-Hb} is normalized. We have (omitting the superscript $(a)$ for clarity)
\[
  \inprod{\psi_{E'}(x), \psi_E(x)} = \int_{- \infty}^{+ \infty} \psi_{E'}^*(x) \psi_E(x) dx
                                   = \delta(E'-E),
\]
which follows from the well-known identity (for instance to be found in~\cite{Cohen-Tannoudji-77})
\[
  \int_{- \infty}^{+ \infty} e^{ik(x-x_0)} dk = 2 \pi \, \delta(x-x_0).
\]
In summary, we have

\begin{theorem}
  In the Wigner quantization of the Hamiltonian $\hat{H}_b = \hat{x} \hat{p}$, the wave function of the particle with position coordinate $x$, when the total energy of the system equals $E$ is given by
\[
  \psi_E^{(a)}(x) =   A \, {\cal{C}}_0(E) \, \frac{|x|^{iE-\frac{1}{2}}}{2 \sqrt{\pi}}
                    + B \, {\cal{C}}_1(E) \, \frac{x|x|^{iE-\frac{3}{2}}}{2 \sqrt{\pi}},
\]
with $|{\cal{C}}_0|^2 = |{\cal{C}}_1|^2 = 1$ and $|A|^2 + |B|^2 = 1$. \\
The wave function of the particle with position coordinate $x$, when the system's momentum is equal to $p$ is given by
\[
  \phi_p^{(a)}(x) = \frac{|xp|^{a-\frac{1}{2}}}{2^a \, \Gamma(a)} 
  \left( \Fser{0}{1}{-}{a}{-\frac{x^2 p^2}{4}} 
         + \frac{ixp}{2a} \, \, \Fser{0}{1}{-}{a+1}{-\frac{x^2 p^2}{4}} \right).
\]
Both results are compatible with the known expressions for these wave functions in canonical quantization, which occurs when $a=1/2$.
\end{theorem}
The wave function in the momentum basis is a superposition of two independent wave functions determined by \eqref{wp-u0E} and \eqref{wp-u1E}.


\section{The Hamiltonian of the free particle $\hat{H}_f = \frac{1}{2} \hat{p}^2$}
Curiously, although many one-dimensional Hamiltonians have been studied before in the context of Wigner quantization, the simplest of them all had been forgotten until now. We choose to fill this lacuna next to the Berry-Keating-Connes Hamiltonian because, concerning Wigner quantization, there are a lot of similarities between $\hat{H}_b$ and the Hamiltonian of the free particle $\hat{H}_f$. Since most of the proofs and methods are very much alike, we will not mention any calculations unless they differ significantly from analogous computations in the previous section.

\subsection{Relation with the $\mathfrak{osp}(1|2)$ Lie superalgebra}
The system of a particle having no potential energy, is described by the Hamiltonian
\begin{equation} \label{ham-free-part}
  \hat{H}_f = \frac{1}{2} \hat{p}^2,
\end{equation}
where $\hat{p}$ is the momentum operator of the particle. We note that the Hamiltonian is independent of the position operator $\hat{x}$. Performing the Wigner quantization for the free particle starts with writing down Hamilton's equations and the equations of Heisenberg for this system. Hamilton's equations give
\[
  \dot{\hat{p}} = - \frac{\partial \hat{H}_f}{\partial x} = 0, \qquad
  \dot{\hat{x}} =   \frac{\partial \hat{H}_f}{\partial p} = \hat{p}.
\]
Together with the equations of Heisenberg (for $\hbar=1$)
\[
  [\hat{H}_f, \hat{p}] = - i \dot{\hat{p}}, \qquad
  [\hat{H}_f, \hat{x}] = - i \dot{\hat{x}}
\]
we obtain a set of compatibility conditions
\[
  [\hat{H}_f, \hat{p}] = 0, \qquad
  [\hat{H}_f, \hat{x}] = - i \hat{p},
\]
which is equivalent to
\begin{equation} \label{CCs-free-particle}
  [\hat{p}^2, \hat{x}] = -2i \hat{p}.
\end{equation}
Just as for the Berry-Keating-Connes Hamiltonian, the operators $\hat{x}$ and $\hat{p}$ subject to the current compatibility conditions \eqref{CCs-free-particle} generate the Lie superalgebra $\mathfrak{osp}(1|2)$. This is not so straightforward to prove as before. 

\begin{theorem}
  If the operators $\hat{x}$ and $\hat{p}$, subject to the relation $[\hat{p}^2, \hat{x}] = -2i \hat{p}$, are considered to be odd elements of some superalgebra, then they must generate the Lie superalgebra $\mathfrak{osp}(1|2)$.
\end{theorem}

\begin{proof}
  We introduce the operators $b^+$ and $b^-$ as before, see equation \eqref{b+b-}, by
\[
  b^\pm = \frac{\hat{x} \mp i \hat{p}}{\sqrt{2}}.
\]
The aim of this proof is to show that $b^+$ and $b^-$ satisfy the defining relations \eqref{osp_def_rel} of $\mathfrak{osp}(1|2)$:
\[
  \left[ \{b^-, b^+\}, b^\pm \right] = \pm 2b^\pm.
\]

First we want to rewrite the compatibility conditions \eqref{CCs-free-particle} in terms of $b^+$ and $b^-$. For this purpose, we start by noticing that
\begin{equation} \label{CCs-b+b-}
  [\hat{p}^2, b^\pm] = b^+ - b^-.
\end{equation}
Either of these two relations can be fully rewritten in terms of $b^+$ and $b^-$. The result is the same, so let us focus on the commutator of $\hat{p}^2$ and $b^+$. Substituting
\[
  \hat{p}^2 = -\frac{1}{2} \bigl( (b^+)^2 + (b^-)^2 - \{ b^+, b^- \} \bigr)
\]
in this relation, we obtain
\[
  [\{ b^+, b^- \}, b^+] - [(b^-)^2, b^+] = 2(b^+ - b^-).
\]
Writing down the commutators and anticommutators explicitly, it is straightforward to see that
\[
  [(b^-)^2, b^+] = - [\{ b^+, b^- \}, b^-],
\]
so that finally both equations in \eqref{CCs-b+b-} are equivalent to 
\begin{equation} \label{CCs-hulpvgl}
  [\{ b^+, b^- \}, b^+ + b^-] = 2(b^+ - b^-).
\end{equation}

Next, assume that $b^+$ and $b^-$ are odd elements of a Lie superalgebra $L$, that is $b^\pm \in L_{\bar{1}}$, where $L_{\bar{1}}$ is the odd part of the Lie superalgebra $L = L_{\bar{0}} \oplus L_{\bar{1}}$. Anticommutators of the odd elements are situated in the even part $L_{\bar{0}}$ of $L$. It is then necessary that the commutator of an odd and an even element must be an odd element again. Thus we can put
\begin{equation} \label{hulprel1}
  [\{ b^+, b^- \}, b^+] = - [(b^+)^2, b^-] = Ab^+ + Bb^-
\end{equation}
for some constants $A$ and $B$. The first equality again follows from explicit computation of the commutators and anticommutators. With the help of equation \eqref{CCs-hulpvgl} we then find:
\begin{equation} \label{hulprel2}
  [\{ b^+, b^- \}, b^-] = - [(b^-)^2, b^+] = (-A+2) b^+ - (B+2) b^-.
\end{equation}
We are able to determine the constants $A$ and $B$ by calculating $[(b^+)^2, (b^-)^2]$ in two different ways. Indeed, from 
\begin{align*}
  [(b^+)^2, (b^-)^2] & = b^+ [b^+, (b^-)^2] + [b^+, (b^-)^2] b^+ \\
                     & = -2 (A-2) (b^+)^2 - (B+2) \{ b^+, b^- \}
\end{align*}
and 
\begin{align*}
  [(b^+)^2, (b^-)^2] & = b^- [(b^+)^2, b^-] + [(b^+)^2, b^-] b^- \\
                     & = -2B (b^-)^2 - A \{ b^+, b^- \}
\end{align*}
we can conclude that $A=2$ and $B=0$. Substituting this into \eqref{hulprel1} and \eqref{hulprel2}, we obtain the $\mathfrak{osp}(1|2)$ defining relations \eqref{osp_def_rel}.
\end{proof}

Thus, the operators $\hat{x}$ and $\hat{p}$ subject to the compatibility conditions \eqref{CCs-free-particle}, have solutions in terms of $\mathfrak{osp}(1|2)$ generators $b^+$ and $b^-$. Clearly $\hat{x}$ and $\hat{p}$ have the same expression as before, and the Hamiltonian $\hat{H}_f$ can be written as
\[
  \hat{H}_f = -\frac{1}{4} \bigl( (b^+)^2 + (b^-)^2 - \{ b^+, b^- \} \bigr)
\]
Moreover, since $\hat{x}$ and $\hat{p}$ together with the compatibility conditions generate $\mathfrak{osp}(1|2)$, no other Lie superalgebra solutions exist.

\subsection{Energy spectrum of the free particle}
In this section, we will show that the spectrum of the Hamiltonian \eqref{ham-free-part} is equal to $\mathbb{R}^+$ with double multiplicity. The Hamiltonian of the free particle written in terms of the $\mathfrak{osp}(1|2)$ generators $b^\pm$ shows that $\hat{H}_f$ is an even operator in $\mathfrak{osp}(1|2)$:
\[
  \hat{H}_f = - \frac{1}{2} (e-f-h),
\]
with $h$, $e$ and $f$ defined in equation \eqref{def_hef}. As we have already argued, an even operator can be regarded in both $\mathfrak{su}(1,1)$ subspaces $V_0$ and $V_1$ separately, which explains how we will end up with a double multiplicity in the energy spectrum. 

Let $z_0(E)$ be a formal eigenvector of $\hat{H}_f$ with eigenvalue $E$ in the subspace $V_0$. We write this eigenvector as
\[
  z_0(E) = \sum_{n=0}^\infty \epsilon_{2n}(E) e_{2n},
\]
where the $\epsilon_{2n}(E)$ are coefficients we wish to determine. We can let $\hat{H}_f$ operate on this eigenvector, which gives
\[
  \hat{H}_f z_0(E) = - \frac{1}{2} \sum_{n=0}^\infty \epsilon_{2n}(E) 
                       \bigl( \sqrt{(n+1)(n+a)} \, e_{2n+2} + \sqrt{n(n+a-1)} \, e_{2n-2} - (2n+a) \, e_{2n} \bigr).
\]
Since $z_0(E)$ is postulated to be an eigenvector of $\hat{H}_f$ with eigenvalue $E$, we also have
\[
  \hat{H}_f z_0(E) = \sum_{n=0}^\infty E \epsilon_{2n}(E) e_{2n}.
\]
Collecting the coefficients of $e_{2n}$ in both expressions for $\hat{H}_f z_0(E)$, delivers us a three term recurrence relation for the unknown functions $\epsilon_{2n}(E)$:
\[
  (2n-2E+a) \, \epsilon_{2n} = \sqrt{n(n+a-1)} \, \epsilon_{2n-2} + \sqrt{(n+1)(n+a)} \, \epsilon_{2n+2}.
\]
This is recognized to be the recurrence relation \eqref{rec_Laguerre} of the normalized Laguerre polynomials with $\alpha = a-1$ and $x=2E$. So we can identify the $\epsilon_{2n}(E)$ with these polynomials:
\begin{eqnarray*}
  \epsilon_{2n}(E) & = & \tilde{L}_n^{(a-1)}(2E) \\
                   & = & \sqrt{\frac{(a)_n}{n!}} \, {\cal{E}}_0(E) \,\, \Fser{1}{1}{-n}{a}{2E},
\end{eqnarray*}
with 
\[
  {\cal{E}}_0(E) = \frac{ E^\frac{a-1}{2} \, e^{-E} \, \sqrt{2^a} }{ \sqrt{\Gamma(a)} }.
\]
Moreover, due to the chosen normalization we have
\[
  \int_0^{+ \infty} \epsilon_{2m}(E) \epsilon_{2n}(E) dE = \delta_{mn}.
\]
The Laguerre polynomials have a weight function with a positive support, so the spectrum of $\hat{H}_f$ in $V_0$ is $\mathbb{R}^+$. In a similar way, one can define a formal eigenvector of $\hat{H}_f$ with eigenvalue $E$ in $V_1$:
\[
  z_1(E) = \sum_{n=0}^\infty \epsilon_{2n+1}(E) e_{2n+1}.
\]
The same procedure as before reveals that 
\begin{eqnarray*}
  \epsilon_{2n+1}(E) & = & \tilde{L}_n^{(a)}(2E) \\
                     & = & \sqrt{\frac{(a+1)_n}{n!}} \, {\cal{E}}_1(E) \,\, \Fser{1}{1}{-n}{a+1}{2t},
\end{eqnarray*}
with 
\[
  {\cal{E}}_1(E) = \frac{ E^\frac{a}{2} \, e^{-E} \, \sqrt{2^{a+1}} }{ \sqrt{\Gamma(a+1)} }.
\]
Thus, we can conclude that the spectrum of $\hat{H}_f$ in $V_1$ is also equal to $\mathbb{R}^+$. Combining both results, we obtain

\begin{theorem}
  The spectrum of $\hat{H}_f = - (e-f-h)/2$ in the representation space $V$ equals $\mathbb{R}^+$ with multiplicity two.
\end{theorem}

Now that we have found the energy eigenstates of $\hat{H}_f$, the corresponding wave function remains to be sought. Crucial for this is the determination of some inner products. 

\subsection{Remaining inner products}
Since $\inprod{v(x), w(p)}$ is already known, the remaining inner products to be found are
\[
  \inprod{v(x), z_0(E)}, \qquad
  \inprod{v(x), z_1(E)}, \qquad 
  \inprod{w(p), z_0(E)}  \qquad \mbox{ and } \qquad 
  \inprod{w(p), z_1(E)}.
\]
The first two are rather straightforward, once the right identity is traced. In~\cite[equation (5.11.3.7)]{Prudnikov-Brychkov-Marichev-86} we find
\begin{equation} \label{eq-IntSeries}
    \sum_{n=0}^\infty \frac{(a)_n}{n!} \, \Fser{1}{1}{-n}{a}{x} \Fser{1}{1}{-n}{a}{y} t^n
  = (1-t)^{-a} \, e^\frac{t(x+y)}{t-1} \, \Fser{0}{1}{-}{a}{\frac{txy}{(1-t)^2}}.
\end{equation}
The eigenvectors of $\hat{x}$ and $\hat{H}_f$ have an inner product that can be found by means of \eqref{eq-IntSeries}. We have
\begin{eqnarray*}
  \inprod{v(x), z_0(E)} & = & \sum_{n=0}^\infty \beta_{2n}(x) \, \epsilon_{2n}(E) \\
                              & = & {\cal{B}}_0(x) {\cal{E}}_0(E) 
                                    \sum_{n=0}^\infty (-1)^n \frac{(a)_n}{n!}
                                                      \, \Fser{1}{1}{-n}{a}{x^2} \Fser{1}{1}{-n}{a}{2E}.
\end{eqnarray*}
Using \eqref{eq-IntSeries} this simplifies to
\begin{equation} \label{vx-z0}
  \inprod{v(x), z_0(E)} = \frac{1}{\sqrt{2^a} \, \Gamma(a)} \, |x|^{a-\frac{1}{2}} \, E^\frac{a-1}{2}
                                \, \Fser{0}{1}{-}{a}{-\frac{2E \, x^2}{4}}. 
\end{equation}
In a similar way, one obtains
\begin{equation} \label{vx-z1}
  \inprod{v(x), z_1(E)} = \frac{1}{\sqrt{2^{a+1}} \, \Gamma(a+1)} \, x |x|^{a-\frac{1}{2}} \, E^\frac{a}{2}
                                \, \Fser{0}{1}{-}{a+1}{-\frac{2E \, x^2}{4}}. 
\end{equation}
The other pair of inner products is a harder nut to crack. We present the highlights of the computation. Omitting a factor of ${\cal{B}}_0(p) {\cal{E}}_0(E)$, the calculation of the inner product $\inprod{w(p), z_0(E)}$ boils down to finding the summation
\begin{equation} \label{sum-wp-z0}
  \sum_{n=0}^\infty \frac{(a)_n}{n!} \, \Fser{1}{1}{-n}{a}{p^2} \Fser{1}{1}{-n}{a}{2E}.
\end{equation}
We apply equation \eqref{lim-2F1} on the hypergeometric function with argument $2E$, before using equations \eqref{eq-Jagannathan} and \eqref{eq-Temme} to simplify the result. After a while, we obtain
\[
  \Gamma(a) \, p^{-2a} \lim_{b \to \infty} \frac{1}{\Gamma(b)} \left( \frac{p^2 b}{2E} \right)^b
                                           \exp {\left( \frac{p^2(2E-b)}{2E} \right)}.
\]
For very large $b$, the Gamma-function can be written as $\sqrt{2 \pi} \, b^{b-1/2} \exp(-b)$, so we can simplify further to reach
\[
  \frac{\Gamma(a)}{\sqrt{2}} \, p^{-2a} \exp(p^2) 
  \lim_{b \to \infty} \sqrt{\frac{b}{\pi}} \exp {\left( b \left[ 1 + \ln(\frac{p^2}{2E}) - \frac{p^2}{2E} \right] \right)}.
\]
At this point, we can introduce the delta function by one of its limit definitions (see for instance~\cite{Cohen-Tannoudji-77})
\[
  \delta(x) = \lim_{\epsilon \to 0^+} \frac{1}{\sqrt{\pi \epsilon}} \exp( -x^2/\epsilon).
\]
Moreover, properties of the delta function allow us to write
\[
    \delta \left( \sqrt{\frac{p^2}{2E} - \ln(\frac{p^2}{2E}) - 1} \right)
  = 2E \, \sqrt{2} \, \delta(p^2-2E),
\]
so that the summation \eqref{sum-wp-z0} can be written as
\[
  \Gamma(a) \, e^{2E} (2E)^{1-a} \, \delta(p^2-2E).
\]
Finding the desired inner product $\inprod{w(p), z_0(E)}$ requires adding a factor ${\cal{B}}_0(p) {\cal{E}}_0(E)$ to the summation \eqref{sum-wp-z0}. Simplifying once again finally gives
\begin{equation} \label{wp-z0}
  \inprod{w(p), z_0(E)} = \sqrt{2} \sqrt{|p|} \, \delta(p^2-2E).
\end{equation}
Analogously, one finds
\begin{equation} \label{wp-z1}
  \inprod{w(p), z_1(E)} = \sqrt{2} \frac{p}{\sqrt{|p|}} \, \delta(p^2-2E).
\end{equation}
At last we are ready to determine the generalized wave functions for the free particle and to compare our results with the canonical case.

\subsection{Generalized wave function and the canonical case}
The free particle has been thoroughly investigated in the past, which makes it easy to check if our results are compatible with what is known. Luckily, they are. We have two wave functions belonging to the same energy eigenvalue $E$. In the position basis, we shall call them $\psi_E^0(x)$ and $\psi_E^1(x)$ and they are defined as the inner products \eqref{vx-z0} and \eqref{vx-z1} respectively. The general wave function of the particle when the system's energy equals $E$, written in the position basis must then be
\begin{equation} \label{wave-function-Hf}
  \psi_E^{(a)}(x) = A \psi_E^0(x) + B \psi_E^1(x),
\end{equation}
with $A$ and $B$ complex coefficients satisfying $|A|^2 + |B|^2 = 1$. This wave function is normalized with respect to the delta function and it is compatible with the canonical case when $a=1/2$. Indeed, for $a=1/2$ equations \eqref{vx-z0} and \eqref{vx-z1} become
\[
  \psi_E^0(x) = \frac{ (2E)^{-\frac{1}{4}} }{ \sqrt{\pi} } \cos(x \sqrt{2E}) 
\]
and
\[
  \psi_E^1(x) = \frac{ (2E)^{-\frac{1}{4}} }{ \sqrt{\pi} } \sin(x \sqrt{2E}).
\]
These are both eigenfunctions with eigenvalue $E$ of the canonical interpretation of the Hamiltonian $\hat{H }_f = - \partial_x^2/2$. Note that, in the canonical case, the solution \eqref{wave-function-Hf} is equivalent to
\[
    \left( \frac{A-iB}{\sqrt{2}} \right) \frac{(2E)^{-\frac{1}{4}} }{ \sqrt{2 \pi} } \, e^{ix \sqrt{2E}}
  + \left( \frac{A+iB}{\sqrt{2}} \right) \frac{(2E)^{-\frac{1}{4}} }{ \sqrt{2 \pi} } \, e^{-ix \sqrt{2E}},
\]
with $|(A-iB)/\sqrt{2}|^2 + |(A+iB)/\sqrt{2}|^2 = 1$. This is the more traditional way of writing the normalized wave function of the free particle, as a superposition of a plane wave moving to the right, and a plane wave going to the left. Note that this canonical notation for the wave function is in accordance with the fact that the energy can be written as $E=p^2/2$ in this case.

\begin{theorem}
  In the Wigner quantization of the Hamiltonian $\hat{H}_f = \frac{\hat{p}^2}{2}$, the wave function of the particle with position coordinate $x$, when the total energy of the system is equal to $E$ is given by
\begin{eqnarray*}
  \psi_E^{(a)}(x) & = &         \frac{A}{\sqrt{2^a} \, \Gamma(a)} \, |x|^{a-\frac{1}{2}} \, E^\frac{a-1}{2}
                                        \, \Fser{0}{1}{-}{a}{-\frac{2E \, x^2}{4}}                                     \\
                  &   & + \quad \frac{B}{\sqrt{2^{a+1}} \, \Gamma(a+1)} \, x |x|^{a-\frac{1}{2}} \, E^\frac{a}{2}
                                        \, \Fser{0}{1}{-}{a+1}{-\frac{2E \, x^2}{4}},
\end{eqnarray*}
with $|A|^2 + |B|^2 = 1$. This result is compatible with the known expression for this wave function in canonical quantization, which occurs when $a=1/2$.
\end{theorem}

In the momentum basis, the wave function is a superposition of the independent wave functions determined by equations \eqref{wp-z0} and \eqref{wp-z1}.

\section{Conclusions}
We have looked at the Wigner quantization of two different one-dimensional systems. The first system is described by the Berry-Keating-Connes Hamiltonian $\hat{H} = \hat{x} \hat{p}$ and in the second part the Hamiltonian of the free particle $\hat{H}_f = \hat{p}^2/2$ is investigated.

Although both systems are entirely different, there are a lot of similarities in their Wigner quantization. Each time solutions for the compatibility conditions were found in terms of generators of the orthosymplectic Lie superalgebra $\mathfrak{osp}(1|2)$. The actions of the operators $\hat{x}$ and $\hat{p}$ are therefore found by looking at the positive discrete series representations of $\mathfrak{osp}(1|2)$, which are characterized by a parameter $a$. We find that the representation space on which all operators act, can be written as a direct sum $V = V_0 \oplus V_1$. Both Hamiltonians also act on this Hilbert space, but their action is confined to one of the two subspaces. The fact that these Hamiltonians have a spectrum with double multiplicity must be attributed to this observation, with the two subspaces $V_0$ and $V_1$ having the same structure.

Likewise, there are some differences to be found when looking at the two Hamiltonians. The spectrum of the Berry-Keating-Connes Hamiltonian covers the entire real axis, while the spectrum of $\hat{H}_f$ only contains positive energy values. The reason for this is that the orthogonal polynomials describing the energy eigenstates have a different support. The support of the weight function of the Meixner-Pollaczek polynomials, appearing for $\hat{H} = \hat{x} \hat{p}$ equals $\mathbb{R}$, while the weight function of the Laguerre polynomials, related to the free particle system, has the positive real axis as its support.

The framework of Wigner quantization is less restrictive than canonical quantization. Therefore, generalized wave functions of the systems can be constructed. These wave functions, two for each Hamiltonian, are expected to be dependent on the representation parameter $a$. Remarkably, for $\hat{H} = \hat{x} \hat{p}$ this is not very much the case. Only the phase factors of the two independent wave functions contain $a$. In contrast with this, we find that the essential structure of the wave functions for the free particle is altered when $a$ changes. Here, we have an actual generalization of the canonical solution.

The latter remark suggests that one is able to retrieve the canonical case in some way. Indeed, one specific representation of $\mathfrak{osp}(1|2)$ corresponds to the canonical picture. For each Hamiltonian we find back the known canonical results for $a=1/2$.

\subsection*{Acknowledgments}
G.\ Regniers was supported by project P6/02 of the Interuniversity Attraction Poles Programme (Belgian State -- 
Belgian Science Policy).


\end{document}